\documentclass[11pt]{article}
\usepackage{amsmath}     
\usepackage{amsthm}
\usepackage{amssymb}
\usepackage[ruled,vlined]{algorithm2e}
\usepackage{mathrsfs}
 \usepackage{psfrag} 			
\usepackage{algorithmic}	
\usepackage{a4wide}
\usepackage{graphicx}

\newtheorem{thm}{Theorem}

\newtheorem{cor}{Corollary}

\newenvironment{keyword}{\par{\noindent\bf Keywords:}}

\begin{document}

\title{Bottleneck combinatorial optimization problems with uncertain costs and the OWA criterion}

\author{Adam Kasperski\thanks{Corresponding author}\\
   {\small \textit{Institute of Industrial}}
  {\small \textit{Engineering and Management,}}
  {\small \textit{Wroc{\l}aw University of Technology,}}\\
  {\small \textit{Wybrze{\.z}e Wyspia{\'n}skiego 27,}}
  {\small \textit{50-370 Wroc{\l}aw, Poland,}}
  {\small \textit{adam.kasperski@pwr.wroc.pl}}
  \and
  Pawe{\l} Zieli{\'n}ski\\
    {\small \textit{Institute of Mathematics}}
  {\small \textit{and Computer Science}}
  {\small \textit{Wroc{\l}aw University of Technology,}}\\
  {\small \textit{Wybrze{\.z}e Wyspia{\'n}skiego 27,}}
  {\small \textit{50-370 Wroc{\l}aw, Poland,}}
  {\small \textit{pawel.zielinski@pwr.wroc.pl}}} 
  
  \date{}
    
\maketitle

\begin{abstract}
In this paper a class of bottleneck combinatorial optimization problems with uncertain costs is discussed. The uncertainty is modeled by specifying a discrete scenario set containing a finite number of cost vectors, called scenarios. In order to choose a solution the Ordered Weighted Averaging aggregation operator (shortly OWA) is applied. The OWA operator generalizes traditional criteria in decision making under uncertainty such as the maximum, minimum, average, median, or Hurwicz criterion. New complexity and approximation results in this area are provided. These results are general and remain valid for many problems, in particular for a wide class of network problems.
 \end{abstract}
 

\section{Introduction.}\label{intro}

In a \emph{combinatorial optimization problem} we are given a finite set of \emph{elements} $E=\{e_1,\dots, e_n\}$ and a set of \emph{feasible solutions} $\Phi\subseteq 2^E$. In a deterministic case, each element $e_i\in E$ has some nonnegative cost $c_i$ and we seek a feasible solution $X\in \Phi$ for which a given cost function $F(X)$ attains minimum. Two types of the cost function are commonly used. The first one, called a \emph{linear sum cost function}, is defined as $F(X)=\sum_{e_i\in X} c_i$, and the second, called a \emph{bottleneck cost function}, is defined as $F(X)=\max_{e_i\in X} c_i$. In this paper, we focus on the bottleneck cost function  and we consider the following \emph{bottleneck combinatorial optimization problem}:
\begin{equation}
\label{pf}
\mathcal{BP}:\; \min_{X\in \Phi} F(X)=\min_{X\in \Phi} \max_{e_i\in X} c_i.
\end{equation}

Formulation~(\ref{pf}) encompasses a large class of problems. In this paper we will consider an important class of network problems, where $E$ is a set of arcs of a given graph $G$ and $\Phi$ contains the subsets of the arcs forming some object in $G$ such as $s-t$ paths (\textsc{Bottleneck Path}), spanning trees (\textsc{Bottleneck Spanning Tree}), $s-t$ cuts (\textsc{Bottleneck Cut}),  or perfect matchings (\textsc{Bottleneck Assignment}). All these network problems are polynomially solvable and were discussed, for example, in~\cite{C78, GT88, P91, P96}.

Because the size of the set $\Phi$ is typically exponential in $n$, we provide with $E$ 
an \emph{efficient description} of $\Phi$, for example a set of constraints that define $\Phi$ or a polynomial time algorithm, which decides whether a given subset of the elements belongs to $\Phi$. We will also consider the following problem $\mathcal{FBP}$ associated with $\mathcal{BP}$:  given a  set of elements $E$ and 
an efficient description of $\Phi$,
check if $\Phi$ is nonempty and if so, return any solution from $\Phi$. For instance,
 in \textsc{Bottleneck Path}, the $\mathcal{FBP}$ problem  consists in checking if there is at least one path between two distinguished nodes of a network encoded in some standard way and if so, returning any such a path. It is easy to verify that $\mathcal{BP}$ can always be reduced to solving at most $n$ problems $\mathcal{FBP}$ and an algorithm works as follows.  We first order the elements of $E$ with respect to nonincreasing costs and remove them one by one in this order, each time solving the resulting 
$\mathcal{FBP}$ problem with the modified description of $\Phi$. We stop when $\Phi$ becomes empty and the last feasible solution enumerated must minimize the bottleneck cost. This method is general and implies that $\mathcal{BP}$ is polynomially solvable if only $\mathcal{FBP}$ is polynomially solvable. We will make use of this idea later in this paper. Let us, however, point out that for some particular problems faster algorithms exist and for descriptions of them we refer the reader to~\cite{C78, GT88, P91, P96}.

The assumption that the element costs are precisely known in advance is often unrealistic. In many cases decision makers can rather provide a set of possible realizations of the element costs. This set is called a \emph{scenario set} and will be denoted by $\Gamma$.  Each vector $\pmb{c}=(c_1,\dots, c_n)\in \Gamma$ is called a \emph{scenario} and represents a realization of the element costs which may occur with a positive but perhaps unknown probability. Several methods of describing scenario set $\Gamma$ were proposed in the existing literature. Among the most popular are the \emph{discrete} and \emph{interval} uncertainty representations~\cite{KY97}. In the discrete uncertainty representation, which will be used in this paper,  $\Gamma=\{\pmb{c}_1,\dots,\pmb{c}_K\}$ contains $K\geq 1$ explicitly given cost scenarios. In the interval uncertainty representation, it is assumed that each element cost may fall within some closed interval and $\Gamma$ is the Cartesian product of all these intervals. In the discrete uncertainty representation, each scenario can model some event that has a global influence on the element costs. On the other hand,
the interval uncertainty representation is appropriate when each element cost may vary within
some range independently on the values of the other costs. Some other methods of defining scenario sets were described, for instance,  in~\cite{AV10,BN09,BS03}.

When scenario set $\Gamma$ contains more than one scenario,  an additional criterion is required to choose a solution. If no probability distribution in $\Gamma$ is provided, then the \emph{robust optimization} framework is widely used~\cite{BN09, KY97}. The idea of robust optimization is to choose a solution minimizing a worst case performance over $\Gamma$. A typical robust criterion is the maximum, i.e. we seek a solution minimizing the largest cost over all scenarios. It turns out that for the linear sum cost function  the minmax problems are typically NP-hard for two scenarios even if the underlying deterministic problem is polynomially solvable. This is the case for the minmax versions of the shortest path~\cite{KY97}, selecting items~\cite{AV01} and other 
classical problems described, for instance, in~\cite{ABV09,KY97}. On the other hand, for the bottleneck cost function the situation is quite different. If the deterministic $\mathcal{BP}$ problem is polynomially solvable, then its minmax version with $K$ scenarios is also polynomially solvable (a straightforward proof of this fact can be found in~\cite{ABV09}). In consequence, the robust bottleneck problems are easier to solve than their corresponding versions with the linear sum cost function. 

The robust (minmax) approach is often regarded as too conservative or pessimistic.  It follows from the fact that the minmax criterion takes into account only one, the worst-case scenario, ignoring the information connected with the remaining scenarios. This fact is well known in decision theory and sample problems in which the minmax criterion seems to be not appropriate can be found, for example, in~\cite{LR57}. 

In this paper, we consider the discrete scenario uncertainty representation, so we assume that the scenario set $\Gamma$ contains $K$ distinct cost scenarios. In order to choose a solution we use the \emph{Ordered Weighted Averaging} aggregation operator (shortly OWA) proposed by Yager in~\cite{YA88}. The key element of the OWA operator are weights whose number equals the number of scenarios. Namely, the $j$th weight expresses an importance of the $j$th largest solution cost over all scenarios. By using various weight distributions one can obtain different criteria such as the maximum, minimum, average, median or Hurwicz. The weights can model an attitude of decision makers towards the risk and allow them to take more information into account while computing a solution. The OWA operator is typically used in multiobjective decision problems (see, e.g.,~\cite{GPS10,GS12,OS03}). However, it is natural to use it also under uncertainty by treating each scenario as a  criterion. 
In this paper we explore the computational properties of the problem of minimizing the OWA criterion in bottleneck combinatorial optimization problems for various weight distributions. 

\section{Problem formulation}

Let scenario set $\Gamma=\{\pmb{c}_1,\dots,\pmb{c}_K\}$ contain $K$ distinct cost scenarios, where $\pmb{c}_j=(c_{1j},\dots,c_{nj})$ for $j\in [K]$ (we use $[K]$ to denote the set $\{1,\dots,K\}$). The bottleneck cost of a given solution~$X$ depends on scenario~$\pmb{c}_j$ and 
will be denoted by
$F(X,\pmb{c}_j)=\max_{e_i\in X} c_{ij}$. Let us introduce weights $w_1,\dots,w_K$ such that $w_j\in [0,1]$ for all $j\in [K]$ and $w_1+\dots+w_K=1$. For a given solution $X$, let $\sigma$ be a permutation of $[K]$ such that $F(X,\pmb{c}_{\sigma(1)})\geq \dots\geq F(X,\pmb{c}_{\sigma(K)})$. The \emph{Ordered Weighted Averaging} aggregation operator (OWA) is defined as follows~\cite{YA88}:
$$\mathrm{OWA}(X)=\sum_{j\in [K]} w_j F(X,\pmb{c}_{\sigma(j)}).$$
The OWA operator has several natural properties which follow directly from its definition. Since it is a convex combination of the cost functions, 
$\min_{j\in [K]} F(X,\pmb{c}_j)\leq \mathrm{OWA}(X) \leq \max_{j\in [K]} F(X,\pmb{c}_j)$. It is also monotonic, i.e. if $F(Y,\pmb{c}_j)\geq F(X,\pmb{c}_j)$ for all $j\in [K]$, then 
$\mathrm{OWA}(Y)\geq \mathrm{OWA}(X)$, idempotent, i.e. if $F(X,\pmb{c}_1)=\dots=F(X,\pmb{c}_K)=a$, then $\mathrm{OWA}(X)=a$ and symmetric, i.e. its value does not depend on the order of scenarios. In this paper we examine the following optimization problem:
$$\textsc{Min-Owa}~\mathcal{BP}:\; \min_{X\in \Phi} \mathrm{OWA} (X).$$

We now discuss several special cases of the OWA operator and the corresponding \textsc{Min-Owa}~$\mathcal{BP}$ problem (see also Table~\ref{tabsc}). If $w_1=1$ and $w_j=0$ for $j=2,\dots,K$, then OWA becomes the maximum and the corresponding problem is denoted as  \textsc{Min-Max}~$\mathcal{BP}$. This is a typical problem considered in the robust optimization framework. If $w_K=1$ and $w_j=0$ for $j=1,\dots,K-1$, then OWA becomes the minimum and the corresponding problem is denoted as \textsc{Min-Min}~$\mathcal{BP}$. In general, if $w_k=1$ and $w_j=0$ for $j\in [K]\setminus\{k\}$, then OWA is the $k$-th largest element and the problem is denoted as \textsc{Min-Quant}$(k)$~$\mathcal{BP}$. In particular, when $k=\lfloor K/2 \rfloor +1$, the $k$-th element is the median and the problem consists in minimizing the median of the costs and is denoted as \textsc{Min-Median}~$\mathcal{BP}$.
If $w_j=1/K$ for all $j\in [K]$, i.e. when the  weights are \emph{uniform}, then OWA is the average (or the Laplace criterion) and the problem is denoted as \textsc{Min-Average}~$\mathcal{BP}$. Finally, if $w_1=\alpha$ and $w_K=1-\alpha$ for some fixed $\alpha\in [0,1]$ and $w_j=0$ for the remaining weights, then we get the \emph{Hurwicz pessimism-optimism} criterion and the problem is denoted as \textsc{Min-Hurwicz}~$\mathcal{BP}$.
\begin{table}
\caption{Special cases of \textsc{Min-Owa}~$\mathcal{BP}$. \label{tabsc}}
\begin{tabular}{ll}
 \hline
	  Name of the problem & Weight distribution \\ \hline
		\textsc{Min-Max}~$\mathcal{BP}$ & $w_1=1$ and $w_j=0$ for $j=2,\dots,K$ \\
		\textsc{Min-Min}~$\mathcal{BP}$ & $w_K=1$ and $w_j=0$ for $j=1,\dots,K-1$ \\
		\textsc{Min-Average}~$\mathcal{BP}$ & $w_j=1/K$ for $j\in [K]$ \\
		\textsc{Min-Quant}$(k)$~$\mathcal{BP}$ & $w_k=1$ and $w_j=0$ for $j\in [K]\setminus \{k\}$ \\
		\textsc{Min-Median}~$\mathcal{BP}$ & $w_{\lfloor K/2 \rfloor +1}=1$ and $w_j=0$ for $j\in [K] \setminus \{\lfloor K/2 \rfloor +1\}$\\
		\textsc{Min-Hurwicz}~$\mathcal{BP}$ & $w_1=\alpha$, $w_K=1-\alpha$, $\alpha\in [0,1]$ and $w_j=0$ for $j\in [K]\setminus\{1,K\}$ \\ \hline
\end{tabular}
\end{table}

\section{Polynomially solvable cases}
\label{secpoly}

In this section, we identify the cases of \textsc{Min-Owa}~$\mathcal{BP}$ which are polynomially solvable. 
It is easy to check that among the problems listed in Table~\ref{tabsc}, \textsc{Min-Max}~$\mathcal{BP}$ and \textsc{Min-Min}~$\mathcal{BP}$ are polynomially solvable if only $\mathcal{BP}$ is polynomially solvable. It is straightforward to check that an optimal solution to \textsc{Min-Max}~$\mathcal{BP}$ can be obtained by computing an optimal solution for the costs $\hat{c}_i=\max_{j\in [K]} c_{ij}$, $e_i\in E$ (see~\cite{ABV09}). This can be done in $O(nK+g(n))$ time, where $g(n)$ is the time required for solving the deterministic $\mathcal{BP}$ problem.
 On the other hand, the  \textsc{Min-Min}~$\mathcal{BP}$ problem can be solved by computing an optimal solution under each scenario and choosing the one with the smallest bottleneck cost, which can be done in $O(Kg(n))$ time.

We now discuss the general problem in which the number of scenarios $K$ is constant (it is not a part of the input). 
\begin{thm}
\label{thm1}
If $\mathcal{FBP}$ is solvable in $f(n)$ time, then \textsc{Min-OWA}~$\mathcal{BP}$ is solvable in $O(n^K (f(n)+K\log K))$ time.
\end{thm}
\begin{proof}
Let $\pmb{F}=(f_1,f_2,\dots, f_K)$ be a vector of, not necessarily distinct, elements of $E$. Element $f_j$ is associated with the scenario $\pmb{c}_j$, $j\in [K]$, and thus vector $\pmb{F}$ defines a vector of costs $\pmb{c_F}=(c'_1,c'_2,\dots,c'_K)$, where $c'_j$ is the cost of $f_j$ under $\pmb{c}_j$. 
Let $\sigma$ be a sequence of $[K]$ such that $c'_{\sigma(1)}\geq c'_{\sigma(2)} \geq \dots \geq c'_{\sigma(K)}$. Define 
${\rm owa}(\pmb{c_F})=\sum_{j\in [K]} w_j c'_{\sigma(j)}.$ 
Let $E_{\pmb{F}}\subseteq E$ be the subset of the elements constructed as follows. For each $j\in [K]$ and $e_i\in E$, if $c_{ij}>c'_{j}$ then $e_i$ is removed from $E$. The set $E_{\pmb{F}}$ contains all the elements which have been not removed.
Let $\Phi_{\pmb{F}}$ be the set of feasible solutions that contain only the elements from $E_{\pmb{F}}$. It is clear that for any $X\in \Phi_{\pmb{F}}$, ${\rm OWA}(X)\leq {\rm owa}(\pmb{c}_F)$. Let $X^*$ be an optimal solution to \textsc{Min-Owa}~$\mathcal{BP}$. We will show that there exists vector $\pmb{F^*}=(f^*_1,f^*_2,\dots, f^*_K)$ such that ${\rm OWA}(X^*)={\rm owa}(\pmb{c_{F^*}})$. Indeed, this equality holds if $f^*_j\in X^*$ is an element of the maximal cost in $X^*$ under scenario~$j$. Now the algorithm works as follows. We enumerate all possible vectors $\pmb{F}$. If the set $\Phi_{\pmb{F}}$ is not empty, then we choose any solution $X_{\pmb{F}}\in \Phi_{\pmb{F}}$. Among the solutions computed we choose the one, say $X_{\pmb{F'}}$, of the minimum ${\rm owa}(\pmb{c_{F'}})$. It is clear that $X_{\pmb{F'}}$ must be optimal.
Let us estimate the running time of the algorithm. The number of vectors $\pmb{F}$ which must be enumerated is bounded by $n^K$. For each $\pmb{F}$ the solution $X_{\pmb{F}}$ can be found by solving the $\mathcal{FBP}$ problem, which requires $f(n)$ time. Finally, computing ${\rm owa}({\pmb{c_F}})$ requires $K\log K$ time.  Hence the running time of the algorithm is bounded by $O(n^K (f(n)+K\log K))$. 
\end{proof}
\begin{cor}
\label{cor1}
If $\mathcal{FBP}$ is polynomially solvable and $K$ is constant, then \textsc{Min-Owa}~$\mathcal{BP}$ is polynomially solvable.
\end{cor}
Of course, the algorithm shown in the proof of Theorem~\ref{thm1} is efficient only when the number of scenarios is small. For large $K$ the result obtained is only theoretical and the approximation algorithm designed in Section~\ref{secappr} might be more appropriate.  We now show how Theorem~\ref{thm1} can be applied to solve the \textsc{Min-Hurwicz}~$\mathcal{BP}$ problem with unbounded $K$.
\begin{thm}
\label{thm2}
If $\mathcal{FBP}$ is  solvable in $f(n)$-time, then \textsc{Min-Hurwicz}~$\mathcal{BP}$ is solvable in $O(Kn^2f(n))$ time.
\end{thm}
\begin{proof}
The Hurwicz criterion with parameter $\alpha \in [0,1]$ can be rewritten as follows:
$${\rm OWA}(X)=\alpha\max_{j\in [K]}\max_{e_i\in X} c_{ij} +(1-\alpha)\min_{j\in [K]} \max_{e_i\in X} c_{ij} =\alpha \max_{e_i\in X} \max_{j\in [K]} c_{ij}+(1-\alpha)\min_{j\in [K]} \max_{e_i\in X} c_{ij}.$$
Let $\hat{c}_i=\max_{j\in [K]} c_{ij}$. Then
$${\rm OWA}(X)=\alpha \max_{e_i\in X} \hat{c}_i + (1-\alpha)\min_{j\in [K]} \max_{e_i\in X} c_{ij} =\min_{j\in [K]} (\alpha \max_{e_i \in X} \hat{c}_i + (1-\alpha)\max_{e_i\in X} c_{ij}).$$
Let us define
$$H_j(X)=\alpha \max_{e_i \in X} \hat{c}_i + (1-\alpha)\max_{e_i\in X} c_{ij},\;\; j\in [K].$$
Hence
$${\rm OWA}(X)=\min_{j\in [K]} H_j(X).$$
Let $X_j$ minimize $H_j(X)$. It is easy to see that an optimal solution to \textsc{Min-Hurwicz}~$\mathcal{BP}$ is among $X_1,\dots,X_K$. So, the problem reduces to solving $K$ auxiliary optimization problems consisting in minimizing $H_j(X)$, $j\in [K]$. But each of these problems is the \textsc{Min-Hurwicz}~$\mathcal{BP}$ with $\alpha$ and only two scenarios, namely $\pmb{c}_1=(\hat{c}_1,\dots,\hat{c}_n)$ and $\pmb{c}_2=(c_{1j},\dots,c_{nj})$, which follows from the fact that $\hat{c}_i\geq c_{ij}$ for each $i,j\in [n]$. Thus, according to Theorem~\ref{thm1}, each $X_j$ can be computed in $O(n^2 f(n))$ time. In consequence, the \textsc{Min-Hurwicz}~$\mathcal{BP}$ problem is solvable in $O(Kn^2f(n))$ time.
\end{proof}
\begin{cor}
\label{cor3a}
If $\mathcal{FBP}$ is polynomially solvable, then \textsc{Min-Hurwicz}~$\mathcal{BP}$ is polynomially solvable.
\end{cor}

We now consider the problem of minimizing the $k$th largest cost, i.e. the problem \textsc{Min-Quant}($k$)~$\mathcal{BP}$. It is clear that this problem is polynomially solvable for $k=1$ and $k=K$ if only $\mathcal{BP}$ is polynomially solvable. We now show that \textsc{Min-Quant}($k$)~$\mathcal{BP}$ is polynomially solvable for any $k$, provided that $k$ is constant.
\begin{thm}
\label{thm3}
		If \textsc{Min-max}~$\mathcal{BP}$ is solvable in $f(n)$ time, then \textsc{Min~Quant}($k$)~$\mathcal{BP}$ is solvable in $O\left(\binom{K}{k-1} f(n)\right)$ time.
\end{thm}
\begin{proof}
	The algorithm works as follows. We enumerate all the subsets of scenarios of size $k-1$. For each such a subset, say $C$, we solve the \textsc{Min-Max}~$\mathcal{BP}$ problem for the scenario set $\Gamma\setminus C$ obtaining a solution $X_C$. Among the computed solutions we return $X_C$ for which the maximal bottleneck cost over $\Gamma\setminus C$ is minimal. It is straightforward to verify that this solution must be optimal. 
\end{proof}
\begin{cor}
		If $k$ or $K-k$ is constant and $\mathcal{BP}$ is polynomially solvable, then \textsc{Min~Quant}($k$)~$\mathcal{BP}$ is polynomially solvable.
\end{cor}
Again, the algorithm proposed in the proof of Theorem~\ref{thm3} is efficient only when $k$ is close to~1 or close to $K$. It has the largest running time when $k=\lfloor K/2 \rfloor +1$, i.e. when we minimize the median of the costs. In this case, however, $k$ is not constant since its value depends on $K$.  In the next section we will show that the \textsc{Min-Median}~$\mathcal{BP}$ problem can be strongly NP-hard and not at all approximable even if $\mathcal{BP}$ is polynomially solvable.

\section{Hard cases}

In order to identify the hard cases of \textsc{Min-Owa}~$\mathcal{BP}$, we have to assume that the number of scenarios $K$ is unbounded (it is a part of the input). Otherwise, \textsc{Min-Owa}~$\mathcal{BP}$ is polynomially solvable (see Theorem~\ref{thm1}). In this section, we first show several negative results when $\mathcal{BP}$ is \textsc{Bottleneck Path}. We then show how these negative results can be extended to other network problems such as
\textsc{Bottleneck Spanning Tree}, \textsc{Bottleneck Cut}, or \textsc{Bottleneck Assignment}.
\begin{thm}
\label{thm4}
If $K$ is unbounded, then \textsc{Min-Average Bottleneck Path} is strongly NP-hard and not approximable within $7/6-\epsilon$ for any $\epsilon>0$ unless P=NP.
\end{thm}
\begin{proof}
We show a polynomial time approximation preserving reduction from the \textsc{Min 3-SAT} problem which is known to be strongly NP-hard~\cite{KM94}. This problem is defined as follows.  We are given boolean variables $x_1,\dots,x_n$ and a collection of clauses $C_1,\dots, C_m$, where each clause is a disjunction of at most three literals (variables or their negations). We ask if there is an assignment to the variables which satisfies at most $L$ clauses. The optimization (minimization) version of the problem is hard to approximate within $7/6-\epsilon$ for any $\epsilon>0$~\cite{AZ02}. Given an instance of \textsc{Min 3-Sat}, we construct the graph shown in Figure~\ref{fig1}. 
\begin{figure}[ht]
\centering
      \includegraphics*{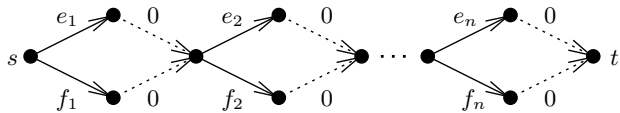}
  \caption{The graph in the reduction.} \label{fig1}
\end{figure}

The arcs $e_1,\dots,e_n$ correspond to literals $x_1,\dots,x_n$ and the arcs $f_1,\dots,f_n$ correspond to literals $\overline{x}_1,\dots,\overline{x}_n$. There is one-to-one correspondence between paths from $s$ to $t$ and assignments to the variables. We fix $x_i=1$ if a path chooses $e_i$ and $x_i=0$ if a path chooses $f_i$.
The set $\Gamma$ is constructed as follows. For each clause $C_j=(l^j_1\vee l^j_2 \vee l^j_3)$,
$j\in [m]$, we form the cost vector $\pmb{c}_j$ in which the costs of the arcs corresponding to
$l^j_1$, $l^j_2$ and $l^j_3$ are set to~$1$ and the costs of the remaining arcs are set to~0. 
Let us fix $w_j=1/m$ for $j\in [m]$.  Suppose that the answer to \textsc{Min 3-Sat} is yes. Then there is an assignment to the variables satisfying at most $L$ clauses. Consider the path $X$ corresponding to this assignment. From the construction of $\Gamma$ it follows that the cost of $X$ is equal to~1 under at most $L$ scenarios. Hence ${\rm OWA}(X)\leq L/m$. On the other hand, if the answer to \textsc{Min 3-sat} is no, then any assignment satisfies more than $L$ clauses and each path $X$ has the bottleneck cost equal to~1 for more than $L$ cost vectors. Hence ${\rm OWA}(X)>L/m$. So, the answer to \textsc{Min 3-SAT} is yes if and only if there is a path $X$ such that ${\rm OWA}(X)\leq L/m$. The reduction is approximation preserving. In consequence, the inapproximability result for \textsc{Min-Average Bottleneck Path} follows.
\end{proof}
\begin{thm}
\label{thm6}
If $K$ is unbounded, then \textsc{Min-Median Bottleneck Path} is strongly NP-hard and not at all approximable unless P=NP.
\end{thm}
\begin{proof}
In order to prove this result we only need to modify the reduction given in the proof of Theorem~\ref{thm4}. We proceed as follows.
	Assume first that  $L<\lfloor m/2 \rfloor$, where $m$ is the number of causes in the instance of \textsc{Min 3-Sat}. We then add to $\Gamma$ additional $m-2L$ scenarios with the costs equal to~1 for all the arcs. So the number of scenarios is $2m-2L$. We fix $w_{m-L+1}=1$ and $w_j=0$ for the remaining scenarios. Now, the answer to \textsc{Min 3-SAT} is yes, if and only if there is a path $X$ whose cost is 1 under at most $L+m-2L=m-L$ scenarios. According to the definition of the weights $\mathrm{OWA}(X)=0$.
	Assume that $L>\lfloor m/2 \rfloor$. We then we add to $\Gamma$ additional $2L-m$ scenarios with the costs equal to~0 for all the arcs.  The number of scenarios is then $2L$. We fix $w_{L+1}=1$ and $w_j=0$ for all the remaining scenarios. Now, the answer to \textsc{Min 3-SAT} is yes, if and only if there is a path $X$ whose cost is 1 under at most $L$ scenarios. According to the definition of the weights $\mathrm{OWA}(X)=0$. We thus can see that it is NP-hard to check whether there is a path $X$ such that ${\rm OWA}(X)\leq 0$ and the theorem follows.
\end{proof}
Theorem~\ref{thm6} leads to the following corollary:
\begin{cor}
\label{cor2a}
If both $K$ and $k$ are unbounded (or $k$ is a function of $K$), then \textsc{Min-Quant}($k$)\textsc{ Bottleneck Path} is strongly NP-hard and not at all approximable unless P=NP.
\end{cor}

The weights used in the OWA operator can reflect an attitude of decision maker towards the risk.  If the weights are nonincreasing, i.e. $w_1\geq w_2\geq \dots \geq w_K$, then decision maker is risk averse. In the extreme case, when $w_1=1$,  he minimizes the largest solution cost over all scenarios. Nonincreasing weights are compatible with the robust optimization framework.  On the other hand, if decision maker is risk seeking, then the weights should be nondecreasing i.e. $w_1\leq w_2\dots \leq w_K$. In the extreme case, when $w_K=1$, he minimizes the smallest solution cost over all scenarios. Theorem~\ref{thm4} implies that \textsc{Min-Owa Bottleneck Path} with both nonincreasing and nondecreasing weights is strongly NP-hard. There is, however, a fundamental difference between these two types of weight distributions. In the next section we will show that the problem with nonincreasing weights can be approximated within some guaranteed worst case ratio. On the other hand, the problem with nondecreasing weights is not at all approximable and this fact is demonstrated by the following result.
\begin{thm}
\label{thm9}
If $K$ is unbounded and the weights are nondecreasing, then \textsc{Min-Owa Bottleneck Path} is strongly NP-hard and not at all approximable unless P=NP.
\end{thm}
\begin{proof}
It is enough to slightly modify the proof of Theorem~\ref{thm4} by replacing the uniform weight distribution with the following one: $w_1=w_2=\dots=w_L=0$ and $w_{L+1}=w_{L+2}=\dots=w_K=1/(K-L)$. If there is an assignment to the variables satisfying at most $L$ clauses, then the cost of the corresponding path $X$ is positive under at most $L$ scenarios and, consequently, $OWA(X)=0$. On the other hand, if each assignment satisfies more than $L$ clauses, then each path has a positive cost under more than $L$ scenarios and $OWA(X)>0$. So, it is NP-hard to check whether there is a path $X$ such that $OWA(X)\leq 0$ and the theorem follows.
\end{proof}
Theorems~\ref{thm4},~\ref{thm6} and~\ref{thm9} and Corollary~\ref{cor2a} hold  when $\mathcal{BP}$ is \textsc{Bottleneck Path}. However, the simple structure of the network in Figure~\ref{fig1} allows us to extend these results for other 
classical network problems. We can see at once, that each optimal $s$-$t$ path in $G$ (see Figure~\ref{fig1}) can easily be transformed into an optimal spanning tree in $G$, and vice versa, by adding or removing a number of dummy (dashed) arcs. Consequently, all the results obtained in this section hold when $\mathcal{BP}$ is \textsc{Bottleneck Spanning tree}. In order to modify the proofs for the \textsc{Bottleneck Cut} and \textsc{Bottleneck Assignment} problems it is enough to replace the graph
 presented in Figure~\ref{fig1} with the corresponding graphs shown in Figure~\ref{fig1a}.
\begin{figure}[ht]
\centering
      {\includegraphics*{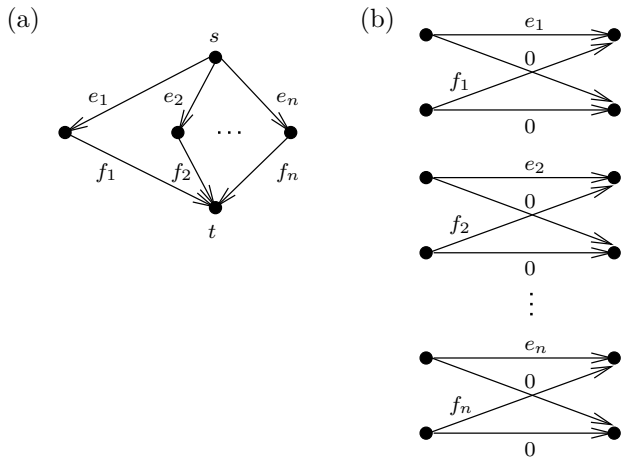}}
     \caption{The graphs: (a) for the minimum $s$-$t$  cut problem,
       (b) for the minimum assignment problem. \label{fig1a}}
\end{figure}

\section{Approximation algorithm} 
\label{secappr}
If the number of scenarios $K$ is large, then the exact algorithm proposed in Section~\ref{secpoly} is inefficient. Furthermore, if $K$ is unbounded, then \textsc{Min-Owa}~$\mathcal{BP}$ can be strongly NP-hard even if $\mathcal{BP}$ is polynomially solvable. In this section, we propose an approximation algorithm with some guaranteed worst-case performance ratio, which can be applied in some cases when $K$ is large. We start with proving the following result:
\begin{thm}
\label{thm7}
		Suppose that $w_1=\dots =w_{k-1}=0$ and $w_k>0$ and
		let $\hat{X}$ be an optimal solution to the \textsc{Min-Quant}($k$)~$\mathcal{BP}$ problem. Then for each $X\in \Phi$, it holds $OWA(\hat{X})\leq (1/w_k)OWA(X)$ and the bound is tight
\end{thm}
\begin{proof}
Let $\sigma$ be a sequence of $[K]$ such that $F(\hat{X},\pmb{c}_{\sigma(1)})\geq \dots \geq F(\hat{X},\pmb{c}_{\sigma(K)})$ and $\rho$ be a sequence of $[K]$ such that $F(X,\pmb{c}_{\rho(1)})\geq \dots \geq F(X,\pmb{c}_{\rho(K)})$. It holds:
$${\rm OWA}(\hat{X})=\sum_{j=k}^K w_j F(\hat{X},\pmb{c}_{\sigma(j)})\leq F(\hat{X},\pmb{c}_{\sigma(k)}).$$
From the definition of $\hat{X}$ and the assumption that $w_k>0$ we get
$$F(\hat{X},\pmb{c}_{\sigma(k)})\leq F(X,\pmb{c}_{\rho(k)})\leq \frac{1}{w_k}\sum_{j \in [K]} w_j F(X,\pmb{c}_{\rho(j)})=\frac{1}{w_k}{\rm OWA}(X).$$
Hence ${\rm OWA}(\hat{X})\leq (1/w_k){\rm OWA}(X)$.
To see that the bound is tight consider 
an instance of
the problem  shown in Table~\ref{tab1}, where $E=\{e_1,\dots,e_n\}$, $n=2K$, $\Phi=\{X\subseteq E: |X|=K\}$ and $w_j=1/K$ for each $j\in [K]$. If we solve the \textsc{Min-Quant}(1)~$\mathcal{BP}$ problem (i.e. the \textsc{Min-Max}~$\mathcal{BP}$ problem) we can get any solution $X\in \Phi$, which follows form the fact that $\max_{j\in [K]} c_{ij}=K$ for each $e_i\in E$.
\begin{table}[ht]
\caption{A hard example for the approximation algorithm.}\label{tab1}
\centering
\begin{tabular}{l|lllllll}
   & $\pmb{c}_1$ & $\pmb{c}_2$ & $\pmb{c}_3$ & $\dots$ & $\pmb{c}_K$ \\ \hline
$e_1$ & 0 & 0 & 0 & $\dots$ & $K$ \\
$e_2$ & 0 & 0 & 0 & $\dots$ & $K$ \\
$\vdots$ \\
$e_K$ & 0 & 0 & 0 & $\dots$ & $K$ \\ \hline
$e_{K+1}$ & $K$ & 0 & 0 & $\dots$ & 0 \\
$e_{K+2}$ & 0 & $K$ & 0 & $\dots$ & 0\\
$\vdots$ \\
$e_{2K}$ & 0 & 0 & 0 & $\dots$ & $K$
\end{tabular}
\end{table}
 Let us choose $\hat{X}=\{e_{K+1},\dots,e_{2K}\}$ with ${\rm OWA}(\hat{X})=K$.  But, when $X=\{e_{1},\dots,e_{K}\}$, then $\mathrm{OWA}(X)=1$. Hence $\mathrm{OWA}(\hat{X})=K\cdot {\rm OWA}(X)=(1/w_1)\mathrm{OWA}(X)$. 
\end{proof}
The algorithm suggested in Theorem~\ref{thm7} is efficient only when the first positive weight is close to 1 or close to $K$. Only in this case, we can solve  the \textsc{Min-Quant}($k$)~$\mathcal{BP}$ problem efficiently and obtain $\hat{X}$ by using the algorithm proposed in Section~\ref{secpoly}.
We now show several consequences of Theorem~\ref{thm7}. If the first weight, associated with the largest cost, is positive, then in order to obtain an approximate solution, we need to solve the \textsc{Min-Max}~$\mathcal{BP}$ problem, which is polynomially solvable if only $\mathcal{BP}$ is polynomially solvable. We thus get the following corollary:
\begin{cor}
If $\mathcal{BP}$ is polynomially solvable and $w_1>0$, then \textsc{Min-Owa}~$\mathcal{BP}$ is approximable within $1/w_1$.
\end{cor}
In the previous section we have shown that the problem with nondecreasing weights can be not at all approximable. The situation is quite different if the weights are nonincreasing, because in this case we have $w_1\geq 1/K$ and we get the following two corollaries.
\begin{cor}
If the weights are nonincreasing and $\mathcal{BP}$ is polynomially solvable, then \textsc{Min-Owa}~$\mathcal{BP}$ is approximable within $1/w_1\leq K$.
\end{cor}
\begin{cor}
If $\mathcal{BP}$ is polynomially solvable, then \textsc{Min-Average}~$\mathcal{BP}$ is approximable within~$K$.
\end{cor}

In general, if the weights are nonincreasing, then the problem is approximable within some guaranteed worst-case ratio not greater than $K$. This ratio takes the largest value, equal to $K$, when the weights are uniform. Observe that the less uniform  the weight distribution  is,  the better 
 worst-case approximation ratio is. The worst case ratio becomes~1 when $w_1=1$, i.e. for the maximum criterion.
 Notice that an approximation algorithm with a worst-case ratio being a polynomially computable function of $K$ cannot exist for the general \textsc{Min-Owa}~$\mathcal{BP}$ problem unless P=NP. This is a consequence of Theorems~\ref{thm6} and~\ref{thm9}, where strong inapproximability results were shown.

\section{Summary}

In this paper we have shown some computational properties of the problem of minimizing the OWA criterion for the class of bottleneck combinatorial optimization problems with uncertain costs. The positive and negative results obtained are general and remain valid for many particular problems. We believe that some of them could be refined by taking into account a particular structure of the underlying problem $\mathcal{BP}$. The computational properties of \textsc{Min-Owa}~$\mathcal{BP}$, where $\mathcal{BP}$ is \textsc{Bottleneck Path}, \textsc{Bottleneck Spanning Tree}, \textsc{Bottleneck Cut}, or \textsc{Bottleneck Assignment}, are summarized in Table~\ref{tabbot}. The most interesting open problem is to close the approximability gap for the average criterion, since the problems are approximable within~$K$ and not approximable within $8/7-\epsilon$ for any $\epsilon>0$ unless P=NP. 
\begin{small}
\begin{table}[ht]
\caption{Summary of the known and new results for the \textsc{Min-Owa}~$\mathcal{BP}$, when $\mathcal{BP}$ is \textsc{Bottleneck Path}, \textsc{Bottleneck Spanning Tree}, \textsc{Bottleneck Cut}, or \textsc{Bottleneck Assignment}; $f(n)$ is the
time required for solving $\mathcal{FBP}$ and $g(n)$ is the time required for solving~$\mathcal{BP}$. }\label{tabbot}

\begin{tabular}{l|l|ll}  
\hline 
		Problem & $K\geq 2$ constant & $K$ unbounded \\ \hline
\textsc{Min-Owa}~$\mathcal{BP}$ &polynomially solvable & strongly NP-hard;  \\
                               &in $O(n^Kf(n))$ time  &  approximable within $1/{w_1}$ if the weights\\ 
															 &        &  are nonincreasing;\\                  
															 &       & not at all 
approximable if the weights \\
															 &       & are arbitrary (or nondecreasing)                 \\\hline \hline
\textsc{Min-Max}~$\mathcal{BP}$  & polynomially  solvable & polynomially solvable   \\
																  & in $g(n)$ time & in $O(Kn+g(n))$ time \\
                                   \hline
\textsc{Min-Min}~$\mathcal{BP}$  & polynomially solvable  & polynomially solvable  \\ 
																 & in $g(n)$ time & in $O(Kg(n))$ time \\ \hline
\textsc{Min-Average}~$\mathcal{BP}$   & polynomially solvable  & strongly NP-hard;  \\ 
																		&		in $O(n^Kf(n))$ time						&approximable within $K$ \\ 
																		&                   & not approximable within $8/7-\epsilon$ \\ \hline
\textsc{Min-Hurwicz}~$\mathcal{BP}$  &polynomially solvable & polynomially solvable  \\
																		  & in $O(n^2 f(n))$ time & in $O(Kn^2f(n))$ time \\
																	  \hline
\textsc{Min-Quant}$(k)$~$\mathcal{BP}$  &polynomially solvable & polynomially
 solvable for constant $k$ \\
                                         &  in $O(n^Kf(n))$ time  					&  (or $K-k$) in 
                                         $O\left(\binom{K}{k-1}(Kn+g(n))\right)$ time;\\
																		      &   & strongly NP-hard and \\ 
																					&   & not at all approximable if $k$ is unbounded \\ \hline
\textsc{Min-Median}~$\mathcal{BP}$ & polynomially solvable & strongly NP-hard; \\
																	&   in $O(n^Kf(n))$ time  & not at all approximable \\ \hline
 \end{tabular}
\end{table}
\end{small}

\bibliographystyle{abbrv} 
\bibliography{owa} 

\end{document}